%% file: A_Sharma_3581-250108.tex
\documentclass[onecolumn, oneside, 11pt, A4]{article}
\input{header.tex}

\def\alf{\frac{a+b}{2}}

\usepackage{amsmath}
\usepackage{amssymb}
\usepackage{psfrag}
\usepackage{epsfig}
\usepackage{fancyhdr}
\usepackage{color}

\usepackage{amsthm}
\newtheorem{thm}{Theorem}[section]

\newtheorem{lem}[thm]{Lemma}
\theoremstyle{remark}

\begin{document}

\title{Model Reduction of Turbulent Fluid Flows\\ Using the Supply Rate}

\author{A. S. Sharma\protect{\footnote{Departments of Electrical and Electronic Engineering; Aeronautics, South Kensington Campus, Imperial College London, London SW7 2BY, UK;\emph{ ati@imperial.ac.uk}}}}
\maketitle

\begin{abstract}
A method for finding reduced-order approximations of turbulent flow models is presented. The method preserves bounds on the production of turbulent energy in the sense of the $\curly{L}_2$ norm of perturbations from a notional laminar profile. This is achieved by decomposing the Navier-Stokes system into a feedback arrangement between the linearised system and the remaining, normally neglected, nonlinear part. The linear system is reduced using a method similar to balanced truncation, but preserving bounds on the supply rate. The method involves balancing two algebraic Riccati equations. The bounds are then used to derive bounds on the turbulent energy production.
An example of the application of the procedure to flow through a long straight pipe is presented. Comparison shows that the new method approximates the supply rate at least as well as, or better than, canonical balanced truncation.
\end{abstract}

\vspace{1cm}

\section{Introduction}
Classically, the study of hydrodynamic stability has proceeded via linearisation and subsequent modal analysis, where the asymptotic behaviour of small perturbations to a steady or slowly-varying base flow are considered.
This asymptotic behaviour is determined by the eigenvalues of a linear operator arising from the analysis, which describes the time-evolution of the eigenmodes.

In many situations this approach gives the whole story. However it relies on the perturbations remaining small enough in the near-term for the asymptotic behaviour to be manifested before the neglected nonlinearities become important.

For the case of flow through a straight pipe of infinite extent examined here, this assumption is not the case. For example dependence of transition on the disturbance amplitude -- a signature of nonlinearity --  was observed by Reynolds in pipe flow a century ago.
Furthermore, the discovery of full solutions other than the simple laminar flow \cite{Waleffe03, Wedin04} necessitates that the laminar solution is only locally stable.

In this case the system is non-normal or non-orthogonal -- indicating that the eigenmodes of the linearised system are not orthogonal with respect to the perturbation energy norm. Specifically, an operator $A$ with adjoint $A^*$ is non-normal if $AA^*\neq A^*A$. In such systems, `interference' or phase effects of the nearly-parallel eigenmodes may result in growth in the norm of interest that in the near-term may dominate the asymptotic behaviour to be expected by merely examining the eigenvalues.
This growth may be large enough for nonlinear effects to become significant in the near-term, rendering the purely classical eigenvalue analysis of secondary importance.
Such non-normality arises in forced, damped systems such as perturbations to a steady viscous shear flow.
This effect is well understood and a review of various ways of understanding non-normality in hydrodynamic systems is to be found in \cite{Schmid07}.

It is also understood that the production of perturbation or `turbulent' energy is completely described by the non-normal dynamics of the linear system, because the forcing arising from the nonlinearity of the perturbed system acts orthogonally to the velocity field and so does not contribute directly to the turbulent energy growth. This indicates that the special choice of the perturbation energy as a norm of interest allows us, in certain respects, to neglect the nonlinearity. This point is made explicit in this paper.

The present work considers the behaviour of arbitrary perturbations to a steady flow solution (laminar flow). The flow is assumed to have sufficiently high viscosity that \emph{small} perturbations eventually decay (so the eigenvalues of the linearised system are all stable), but sufficiently low viscosity that the same small perturbations may be capable of growing in the intermediate term (the linearised system is non-normal). This work may therefore be understood in the broader context of work applying ideas from the field of systems theory to fluids problems. Examples of this approach would include \cite{Jovanovic05, Bamieh01} and \cite{Bewley01}.

It is useful at this stage to consider the system as decomposed into two connected sub-systems, a linear system (obtained by linearisation) and a nonlinear part (the part that would normally be neglected). This view allows their separate analysis and the analysis of their mutual interaction. Furthermore, the linear part may be itself decomposed into various modes. For the pipe geometry in question, an orthogonal wavenumber decomposition is appropriate and further decomposition at each wavenumber according to turbulent energy production or dissipation is suggested.

Turbulence is a self-sustaining process, so the flow perturbations need to extract energy from the base flow, to counter dissipation.
This energy extraction is described in the sequel by a model obtained by linearisation about the laminar profile. A sector-bounding or supply rate analysis is used to pick out the dominant modes or structures in this amplification process. We might expect these structures to feature prominently in turbulent flows, along with the associated dissipative modes or structures excited via the nonlinearity.
The forcing required to drive the linear amplification process is provided throughout the flow volume by the nonlinearity. The nonlinearity is itself driven by the outcome of the linear amplification process, completing a self-supporting feedback loop.
Where the amplification is insufficient, the process will eventually decay.
Furthermore, the most amplifying modes can be calculated and the threshold at which amplification is insufficient to sustain turbulence can be quantified exactly.

In the actual system, as in this simple model, nonlinear effects interact with linear effects. However, given that the nonlinear effects do not create or destroy turbulent energy, but merely redistribute it, we can consider the linear subsystem as both producer and consumer of turbulent energy, and will look to decompose it accordingly. We might hypothesise that structures dominating in turbulent flow will correspond to those linear modes that strongly produce turbulent energy. The structures grow until their linear (locally determinable) growth is no longer sustainable, and the nonlinear forcing becomes significant enough to excite the more dissipative modes, dissipating energy. The nonlinearity meanwhile also provides forcing to the productive linear modes, allowing the process to begin again. The process as a whole depends on an energy source to initiate the growth (the shearing interaction with the base flow), which is manifested in the equations as a non-normality, and a nonlinearity to provide saturation and `clearing' so that the process can begin again.

The system as a whole is then understood as a feedback process between a conservative, redistributive nonlinearity, and energy-producing or energy-dissipating linear modes.
This self-reinforcing feedback process was partly described in \cite{Trefethen93}.

The presented analysis bounds the turbulent energy production of an 
arbitrary flow at a given Reynolds number. In this respect it is a nonlinear analysis. The paper goes on to propose a novel reduced-order approximation of the linear dynamics, that preserves bounds on this turbulent energy production, similar to balanced truncation. For example, see \cite{Moore81, Green} and, in the current context, \cite{Rowley04}. Should the nonlinearity be retained, it is anticipated that the nonlinear dynamics are also approximated. The proposed arrangement is illustrated in Fig.~\ref{reducedloop}.
In the sense that this approach is concerned with the contribution to turbulent energy production, the presented analysis is related to the stochastic forcing model of \cite{Farrell93}, the passivity analysis of \cite{Sharma06AIAA} and the $\mathcal{H}_{\infty}$-norm analysis of \cite{Jovanovic05}. 

\begin{figure}[p]
\psfrag{N}[][]{$\curly{N}$}
\psfrag{G}[][]{$G$}
\psfrag{N}[][]{$\curly{N}$}
\psfrag{Gred}[][]{${G}_{r}$}
\psfrag{v}[][]{$v$}\psfrag{n}[][]{$n$}
\psfrag{modelredn}[][]{sector balanced truncation}
\centering \includegraphics[width=12cm]{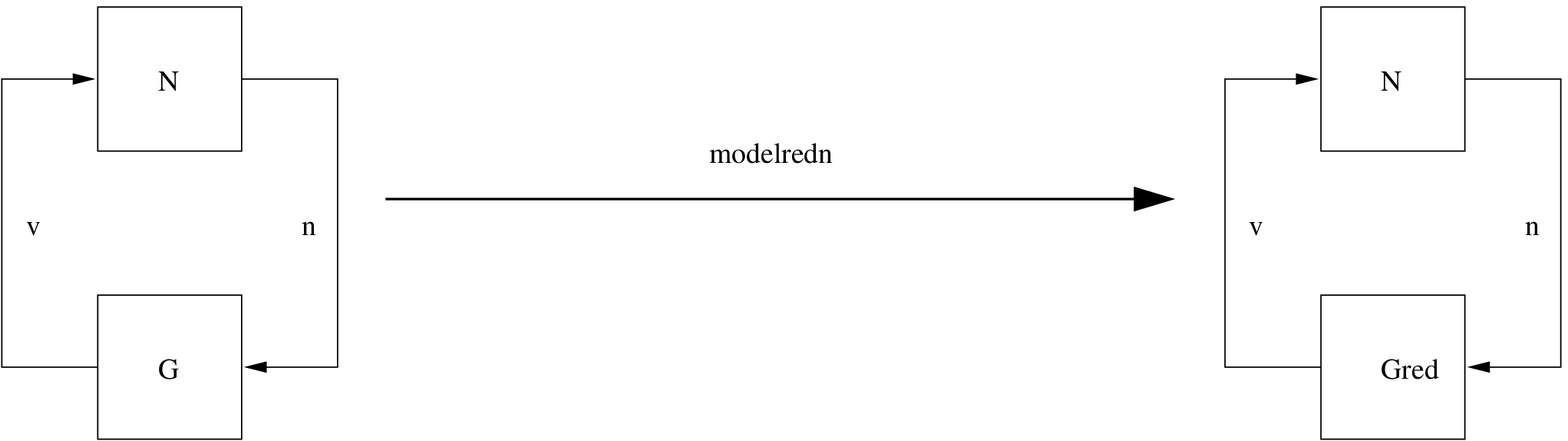}
\caption{The nonlinear model reduction process. $G$ is the full-order linear system, $G_r$ is the reduced-order linear system. $\curly{N}$ is the nonlinearity. The velocity field is $v$ and the forcing provided by the nonlinearity is $n$}
\label{reducedloop}
\end{figure}

The proposed scheme contrasts with the canonical balanced truncation scheme in two important ways. Firstly, the famous `twice-the-sum-of-the-tail' infinity norm error bound \cite{Green} is lost using the new procedure. Secondly, balanced truncation does not guarantee that sector bounds are preserved during the procedure.
Because the linear system to be reduced may exist in connection with a setor-bounded nonlinearity, it may be important for nonlinear stability results to preserve sector bounds during the reduction process. For a simple motivating example, consider the system shown in Fig.~\ref{reducedloop}. If both the nonlinearity $\curly{N}$ and the linear system $H$ in the left half of Fig.~\ref{reducedloop} are known to be strictly passive, the system arising from their interconnection must be stable \cite{Zames66-1}. Imagine balanced truncation is applied to $G$ yielding a reduced order model $G_r$. The canonical balanced truncation procedure is not guaranteed to preserve passivity. Therefore, the interconnection of $G_r$ with $\curly{N}$ shown in the right half of the figure may not itself be stable. Consequently, we might like to preserve sector bounds in any truncation procedure used in the presence of a sector bounded nonlinearity.

\section{Sector Bounded Systems}
We begin briefly reviewing the theory of sector bounded systems. A classic, comprehensive introduction is to be found in \cite{Willems72-0, Willems72-1}.

Consider a linear, time invariant (LTI) system $H$,
\begin{equation}
y = Hu
\end{equation}
with transfer function $H(s)$ that is analytic in the closed right half-plane.

The system $H$ is in sector $[a,b]$ if the transfer matrix $H(i\omega)$ obeys the following inequality
\begin{equation}
\Lambda (i\omega) = \textrm{Herm} [(H(i\omega)-aI)^{*}(H(i\omega)-bI)]\leq 0 \quad \forall \omega \in \mathbb{R}.
\label{matrixsectordefn}
\end{equation}
where we use the Hermitian part of a matrix, $\textrm{Herm}(X)=(X^*+X)/2$.

Equivalently, if there is a supply rate $w(y(t),u(t))$ 
and a non-negative storage function $S(z)$ such that
\begin{equation}
S(z(t)) = S(z(0)) + \int_{t_{0}}^{t} w(y(s),u(s))ds
\end{equation}
with the supply rate
\begin{equation}
w(y(t),u(t)) = (a+b)y^{*}(t)u(t) - abu^{*}(t)u(t) - y^{*}(t)y(t) \geq 0
\label{supplyrate}
\end{equation}
then $H \in [a,b]$, where the state of $H$ is $z$, and $S(0)=0$.

This can be interpreted graphically as the graph of the system lying within a conic region in the input-output space inhabited by $H$ (see Fig.~\ref{sectorbounded} for a single-input, single-output example, and see \cite{Zames66-1}).

\begin{figure}[p]
\psfrag{y}[][]{y}
\psfrag{u}[][]{u}
\psfrag{c+r}[][]{a=c+r}
\psfrag{c-r}[][]{b=c-r}
\psfrag{c}[][]{c}
\centering \includegraphics[width=12cm]{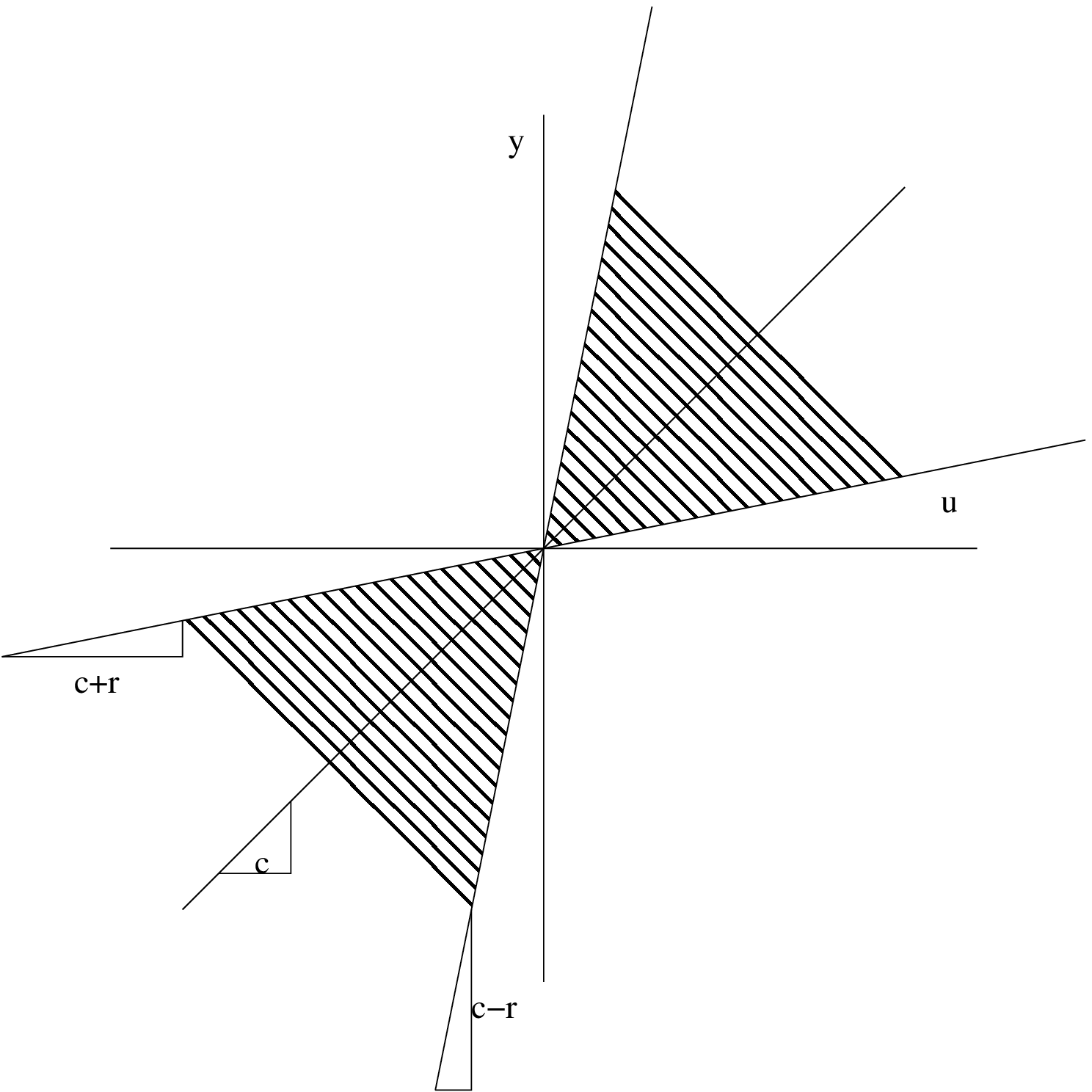}
\caption{For a system $y=Hu$,  $H\in [a,b]$, or $w>0$. The gradients of the two lines are $a$ and $b$. The shaded area represents the possible input-output pairs}
\label{sectorbounded}
\end{figure}

There are two special instances of broader interest, the sector $[0,\infty]$ which corresponds to a passive system (see Fig.~\ref{passive}), and the sector $[-\gamma,\gamma]$ which corresponds to $\Hinfty$-norm bounded systems (see Fig.~\ref{gammabounded}). Clearly the sector bounds $a$ and $b$ are not unique; for example, a system in sector $[-1,1]$ is also in sector $[-2,2]$.
\begin{figure}[p]
\psfrag{y}[][]{y}
\psfrag{u}[][]{u}
\centering \includegraphics[width=12cm]{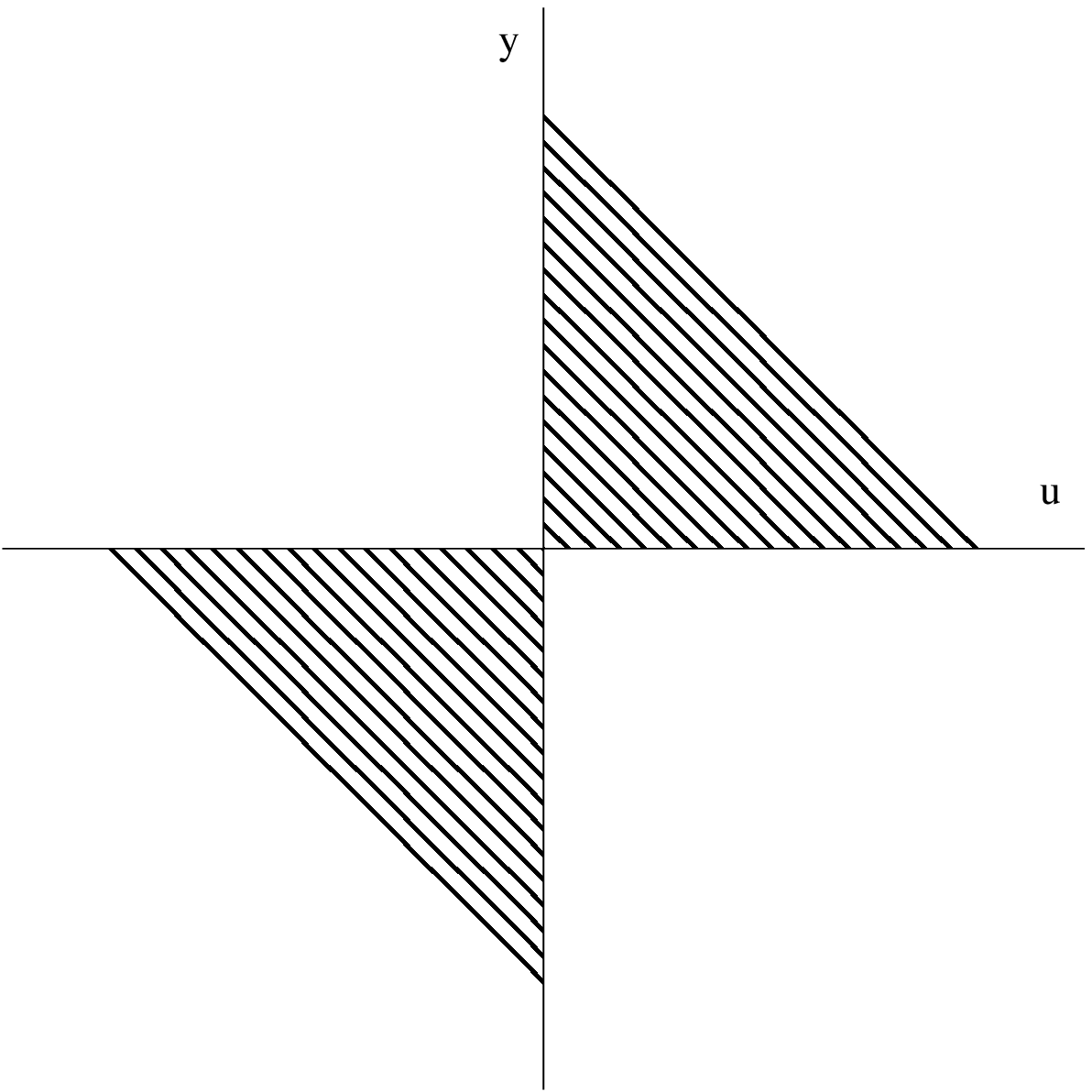}
\caption{Passive systems are a special case of Fig.~\ref{sectorbounded}, where $H\in [0,\infty]$, $H$ is passive, or $y^{*}u\geq 0$. A passive system never produces more output energy than that which is stored within the system plus that which is supplied to it.
}
\label{passive}
\end{figure}

\begin{figure}[p]
\psfrag{y}[][]{y}
\psfrag{u}[][]{u}
\psfrag{\\gamma}[][]{$\gamma$}
\centering \includegraphics[width=12cm]{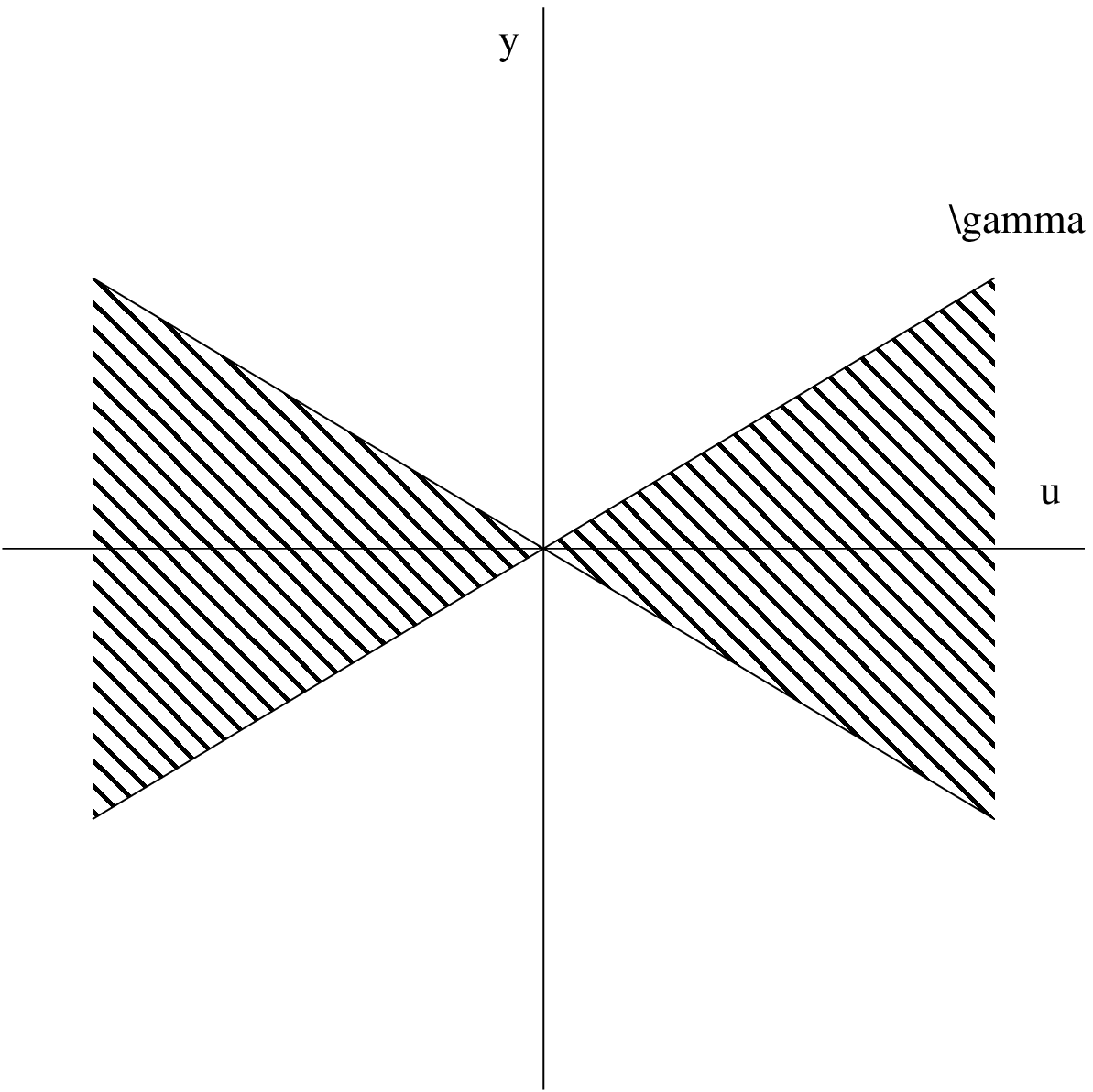}
\caption{\Hinf-norm bounded systems are a special case of Fig.~\ref{sectorbounded}, where $H\in [-\gamma,\gamma]$, or $\norm{y}{2}/\norm{u}{2}<\gamma$, or $\norm{H}{\infty}\leq\gamma$. Such a system's output energy is bounded by a ratio to the input energy.}
\label{gammabounded}
\end{figure}

We will make use of the following properties. If $H_{1}\in[a_1,b_1]$ and $H_{2}\in[a_2,b_2]$, with $b_1$ and $b_2$ positive, then $H_1+H_2\in[a_1+a_2,b_1+b_2]$.
If $H\in[a,b]$ and $k>0$, then $kH\in[ka,kb]$.

\section{Sector Bounding of the Linearised Navier-Stokes Equations and Bounding of the Turbulent Energy Production}
In this section we introduce the linearised, perturbed Navier-Stokes equations. The equations can be sector bounded and these bounds can be used to bound turbulent energy production. The perturbed Navier-Stokes equations can be interpreted as a coupled linear system and memoryless nonlinear system.

We begin by writing down the equations for three-dimensional incompressible fluid flow
evolving in time in a domain $\Omega \subset \bspace{R}^{3}$.
The state of the flow at an instant in time $t$ is fully described by a time-dependent velocity vector field $V(x,t)$ and a scalar pressure field $P(x,t)$ where $x$ is a point in $\Omega$ and $t$ is a point in time.

The flow is governed by the incompressible Navier-Stokes equations
at Reynolds number $Re$. The equations of motion are
\begin{equation}
\label{NS}
\begin{split}
\dot{V}(x,t) =& -V(x,t) \cdot \grad V(x,t) - \grad{P}(x,t) + \frac{1}{Re}
\grad^{2}{V}(x,t)
\end{split}
\end{equation}
\begin{equation}
\label{divfree}
\div V(x,t)=0.
\end{equation}
The flow also obeys prescribed boundary conditions
\begin{equation}
V(x,t) = V_{\partial}(x,t).
\end{equation}

We consider perturbations $v(x,t)$ around an assumed steady solution $\bar{v}(x)$. This gives the net velocity vector field
\begin{equation}
\label{perturb}
V = \bar{v}+v.
\end{equation}
The steady pressure $\bar{p}(x)$ is similarly perturbed  by $p(x,t)$.

Substitution into (\ref{NS}) gives the perturbation equations
\begin{equation}
\label{perturbedNS}
\begin{split}
\dot{v}(x,t) =& - \bar{v}(x,t) \cdot \grad v(x,t)
 - v(x,t) \cdot \grad \bar{v}(x,t)
 + n(x,t)
 - \\
 & \grad{{p}}(x,t) + \frac{1}{Re}
\grad^{2}{v}(x,t) \\
n(x,t) =& -v(x,t) \cdot \grad v(x,t),\\
\div{v}(x,t) =& 0.
\end{split}
\end{equation}
A substitution has been made for the nonlinear part, giving coupled linear and nonlinear equations.
We do not make the assumption of small perturbations.

The equations (\ref{perturbedNS}) include pressure. The pressure term can be eliminated along with the divergence equation by projecting the equations onto the space of divergence-free functions. Defining $\Pi$ via $\Pi(\grad{p})=0$,  $\Pi(v)=v$ gives
\begin{equation}
\label{projNS}
\begin{split}
\dot{v}(x,t) =& - \Pi \left( \bar{v}(x,t) \cdot \grad v(x,t) - v(x,t) \cdot \grad \bar{v}(x,t)  + \frac{1}{Re}
\grad^{2}{v}(x,t) + n(x,t) \right)\\
n(x,t) =& \curly{N}(v)= - v(x,t) \cdot \grad v(x,t).
\end{split}
\end{equation}

\label{finitedim}
The perturbation equations (\ref{projNS}) are then described as the feedback interconnection between a linear part and the nonlinear part.

We achieve this by writing the system equations (\ref{projNS}) in operator form as
\begin{equation}
\begin{split}
\label{opNS}
\dot{{v}}(x,t) &= Av(x,t) + n(x,t)  \\
n(x,t)          &= \curly{N}(v(x,t))
\end{split}
\end{equation}
and for notational convenience write this as
\begin{equation}
\begin{split}
v =& Gn\\
n =& \curly{N}(v).
\end{split}
\end{equation}

Define an inner product on $\Omega\times[t_{0},t]$, 
\begin{equation}
\inprod{\alpha}{\beta} = \int_{t_{0}}^{t}\int_{x\in\Omega}\beta(x,s)^{*}\alpha(x,s)dx ds.
\label{fullinnerprod}
\end{equation}

Using this definition, consider the total turbulent energy, $E(t) = \frac{1}{2} \int_{x\in\Omega} v(x,t)^2 dx$,
\begin{equation}
{E}(t) = \inprod{v}{\curly{N}(v)} + \inprod{Gn}{n}.
\end{equation}
The fields in question are real-valued.
It is easily shown \cite{Sharma06AIAA} that the net inflow of perturbation energy into the spatial domain $\Omega$ is 
\begin{equation}
{E}_{in} = \inprod{v}{\curly{N}(v)}.
\end{equation}
To see this, evaluate
$ \inprod{v}{n}$
or equivalently,
\begin{equation}
\label{nonlinreq}
\inprod{v}{\curly{N}v}=
-\int_{t_{0}}^{T}\int_{x\in \Omega} v(x,t)\cdot \left( v(x,t) \cdot \grad v(x,t)\right) ~dx~dt.
\end{equation}

Applying the divergence theorem, the inner
integral is equivalent to an integral over the boundary,
\begin{equation}
\label{nonlinbc}
-\int_{x\in \partial\Omega} (v(x,t)\cdot v(x,t)) v(x,t) \cdot \hat{\xi} ~dx
\end{equation}
where $\hat{\xi}$ is the outward-facing
unit vector perpendicular to the boundary
$\partial\Omega$. Physically interpreted, (\ref{nonlinbc})
quantifies the net flux of disturbance energy into the domain
through the boundary per unit time.

So, the perturbation energy is
\begin{equation}
{E}(t) = -\inprod{Gn}{n} + {E}_{in}.
\label{pertenergy}
\end{equation}

Assume the linearised system ($G$) is in sector $[a,b]$. We can find a suitable $a$ and $b$ later. Since $G$ is in sector $[a,b]$, the associated supply rate is non-negative.
To evaluate the supply rate, we also need the inner product on $\Omega$ only, which we define via
\begin{equation}
\inprod{\alpha}{\beta}_{\Omega} = \int_{x\in\Omega}\beta^{*}(x)\alpha(x) dx.
\label{innerprod}
\end{equation}
This is the time derivative of Eq.~(\ref{fullinnerprod}).
The supply rate is then
\begin{equation}
w(t) = (a+b)\inprod{v(t)}{n(t)}_{\Omega} - ab \inprod{n(t)}{n(t)}_{\Omega} - \inprod{v(t)}{v(t)}_{\Omega} \geq 0
\end{equation}
which, using the time derivative of Eq.~(\ref{pertenergy}), gives
\begin{equation}
(a+b)(\dot{E}_{in} - \dot{E})-ab\inprod{n(t)}{n(t)}_{\Omega} - \inprod{v(t)}{v(t)}_{\Omega} \geq 0
\end{equation}
where $\dot{x}$ denotes the time derivative of $x$.
To allow for the very viscous small-scale modes we expect, take the limit of large $b$, so that $G\in [a,\infty]$.
Then we obtain a bound on the rate of perturbation (turbulent) energy creation,
\begin{equation}
\dot{E} \leq -a\norm{n(t)}{\Omega}^{2} + \dot{E}_{in}.
\end{equation}
We would seek an approximatation to $G$, retaining modes that contribute maximally to this turbulent energy creation.

The turbulent energy production is bounded from above by the energy inflow, the sector bound $a$, and the norm of the nonlinear forcing function of $v$, $n=-v\cdot \grad v$.

Note that if $a<0$, growth of the turbulent energy is possible, even from small $v$, via the nonlinear feedback mechanism.
If however $a>0$, and there is no net inflow, then the turbulence is guaranteed to decay.

Conservative estimates for the sector bounds $a$ and $b$ may easily calculated from the linearised system by taking limits of Eq.~(\ref{matrixsectordefn}), giving
\begin{equation}
b \geq \frac{1}{2}(G(j\omega)^{*} + G(j\omega) ) \geq a.
\label{abbound1}
\end{equation}
Bounds may also be calculated using the infinity norm (see Fig.~\ref{gammabounded}),
\begin{equation}
\norm{G}{\infty} \leq b, \quad \norm{G}{\infty} \leq -a.
\label{abbound2}
\end{equation}

If $a<0$, the linear part of the system may produce turbulent energy via the non-normality in the perturbed equations, which may be interpreted graphically as in Fig.~\ref{sectorainf}.
\begin{figure}[p]
\psfrag{y}[][]{y}
\psfrag{u}[][]{u}
\centering \includegraphics[width=12cm]{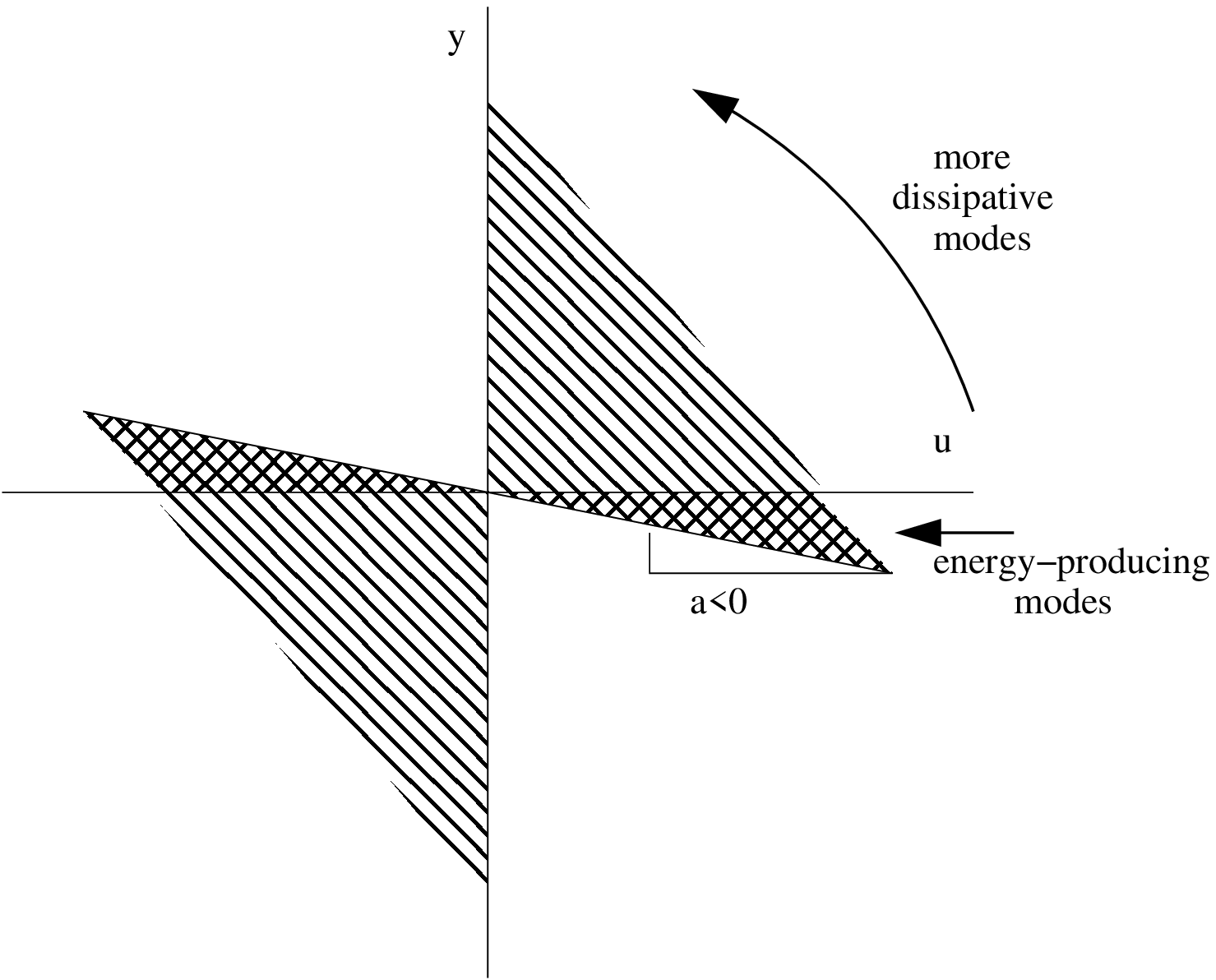}
\caption{$H\in [a,\infty]$. The energy $\norm{y}{2}$ may increase transiently, even though $H$ is stable, but the growth is bounded. In the presence of a conservative nonlinear feedback, the system may not return to rest, but the energy production rate is bounded.}
\label{sectorainf}
\end{figure}

\section{Model Reduction by Sector-Balanced Truncation}
\label{ModelReduction}
The novel reduction procedure presented here is designed to preserve sector bounds and therefore the nonlinear turbulent energy production bounds described in the previous section.
The procedure is in this sense appropriate for application to finding finite-dimensional approximation of Navier-Stokes systems.

The previous section showed how sector bounds can be used to bound the turbulent energy production. In this section we reduce the order of the linearised system responsible for producing the turbulent energy, while leaving the nonlinear part alone.

The proposed scheme contrasts with the canonical balanced truncation scheme in two important ways. Firstly, the famous `twice-the-sum-of-the-tail' infinity norm error bound \cite{Green} is lost. Secondly, balanced truncation does not guarantee that sector bounds are preserved during the procedure.
Because the linear system to be reduced may exist in connection with a sector-bounded nonlinearity, it may be important for nonlinear stability results to preserve sector bounds during the reduction process. For a simple motivating example, consider the system shown in Fig.~\ref{reducedloop}. If both the nonlinearity $\curly{N}$ and the linear system $H$ in the left half of Fig.~\ref{reducedloop} are known to be strictly passive, the system arising from their interconnection must be stable \cite{Zames66-1}. Imagine balanced truncation is applied to $G$ yielding a reduced order model $G_r$. The canonical balanced truncation procedure is not guaranteed to preserve passivity. Therefore, the interconnection of $G_r$ with $\curly{N}$ shown in the right half of the figure may not itself be stable. Consequently, we might like to preserve sector bounds in any truncation procedure used in the presence of a sector bounded nonlinearity.

For an introductory survey of model reduction by various balanced truncation variants, the reader is directed to \cite{Gugercin04}.

To summarise, the procedure is as follows:
\begin{enumerate}
\item Find the linearised system $v=Gn$ and the associated inner product $\inprod{\cdot}{\cdot}_{\Omega}$
\item Establish that $G\in [a,b]$
\item Find an approximation to $G$, $\tilde{G}$, that is also in the same sector.
\end{enumerate}

In the spirit of balanced truncation, we propose a similar state-space truncation that preserves the sector bounding in the next section. We will also see that canonical balanced truncation is recovered in a special limting case.

Suppose the system $v=Hf$ is square, stable, finite-dimensional, linear time-invariant and $H\in [a,b]$, with minimal state-space realisation
\begin{eqnarray}
\dot{x}(t) = Ax(t)+ Bf(t)\\
v(t) = Cx(t)+Df(t)\\
E(t)=v^*(t)v(t)
\end{eqnarray}
with $x \in \bspace{C}^m$, $v\in\bspace{C}^l$, $f\in \bspace{C}^l$ and complex matrices $A,~B,~C,~D$ dimensioned compatibly.
This can be written as
\begin{equation}
H =  \mat{c|c}{A & B \\ 
\hline C  & D}.
\label{ssrealisation}
\end{equation}

If the similarity transformation $T$ is a suitably dimensioned square nonsingular matrix, the following is also a realisation of $H$,
\begin{eqnarray}
\dot{z}(t) = TAT^{-1}z(t)+ TBf(t)\\
\label{transformed1}
v(t) = CT^{-1}z(t)+Df(t)\\
z(t)=Tx(t).
\label{transformed}
\end{eqnarray}

$H$ is strictly in sector $[a,b]$ if and only if the following conditions are met \cite{Shim97},
\begin{enumerate}
\item{There exist real, symmetric solutions $X>0$ and $Y>0$ of the following two Algebraic Riccati equations (AREs)
\begin{equation}
A'X + XA - (XB-C'(\alf I - D))P^{-1}(XB-C'(\alf I - D))' + C'C=0
\label{RiccX}
\end{equation}
\begin{equation}
AY + YA' - (YC'-B(\alf I - D)')P^{-1}(YC'-B(\alf I - D)')' + BB'=0
\label{RiccY}
\end{equation}
}
\item{$P<0$}
\item{The solutions $X$ and $Y$ are stabilising, that is,
\begin{equation}
A - BP^{-1}(\alf I-D)'C - BP^{-1}B'X<0
\label{StabilisingX}
\end{equation}
\begin{equation}
A' - C'P^{-1}(\alf I-D)B' - C'P^{-1}C'Y<0
\label{StabilisingY}
\end{equation}
}
\end{enumerate}
where $P = D'D - \alf(D+D') + abI$.

Since existence of solutions to the AREs (\ref{RiccX}) and (\ref{RiccY}) determine that the system obeys the sector bound $H\in[a,b]$ and thus has a supply rate in Eq.~(\ref{supplyrate}), it is appropriate to determine a balancing and truncation scheme whereby the solutions $X$ and $Y$ are equal and diagonal. The states spanned by the eigenvectors of the balanced $X$ and $Y$ corresponding to their smallest eigenvalues are then truncated. The truncated states are then those that impact least on the supply rate.
Thus we seek a transformation $T$ so that the solutions to the two AREs for the realisation in Eqs.~(\ref{transformed1}-\ref{transformed}) are equal and diagonal with ordered elements.

A transformation that achieves this is $T=\Sigma^{\frac{1}{2}}U'R^{-1}$ where $Y=RR'$ is a Cholesky factorisation of $Y$ and $R'XR = U\Sigma^{2}U'$ is a singular value decomposition of $R'XR$. A proof is presented in Appendix~\ref{apx_balancing}.

The balanced system may then be partitioned,
\begin{equation}
H = \mat{c|c}{\tilde{A} & \tilde{B} \\ \hline \tilde{C} & D} =
\mat{cc|c}{A_{11} & A_{12} & B_{1} \\ 
A_{21} & A_{22} & B_{2} \\
\hline C_{1} & C_{2} & D}.
\label{partitioned}
\end{equation}
The partitioned system is then truncated 
\begin{equation}
H_{r} =  \mat{c|c}{A_{11} & B_{1} \\ 
\hline C_{1}  & D}.
\label{truncated}
\end{equation}
It is also shown in Appendix \ref{apx_balancing}  that the truncated realisation in Eq.~(\ref{truncated}) is itself sector balanced, stable, and $H_{r}$ is also strictly in sector $[a,b]$.

Since the reduced system $H_{r}$ shares the same input and output dimensions as the full system $H$, the same nonlinearity is used in the approximated nonlinear model.

Computationally, it is advisable to relax the sector bounds $a$ and $b$ where they have been found by Eqs (\ref{abbound1}) and (\ref{abbound2}) to improve the conditioning of the AREs.

It is interesting to note that canonical balanced truncation is a special case where the sector under consideration approaches $[-\infty,\infty]$, in which case the AREs reduce to Lyapunov equations and their solutions to the controllability and observability gramians. In this case the method coincides with balanced truncation, and bounds on the turbulent energy production are lost.

\section{Application to Hagen-Poiseuille Flow}
As an illustrative example, we use flow through an infinitely long, straight pipe. The pipe is of particular relevance to this type of analysis because the eigenmodes of its linearised operator are always stable. Consequently we must turn to an analysis of system non-normality to explain turbulent energy production.

The difficulty of finding solutions to AREs involving partial differential operators unfortunately imposes that we begin with a finite-dimensional model before finding the reduced-order model.
As such, the flow is Fourier transformed in the streamwise (axial) and azimuthal directions. The highly accurate linear code presented in \cite{Meseguer03-1} is used. The flow is discretised with a discretisation resolution $N=100$.
The Reynolds number used is defined as
\begin{equation}
Re=\frac{\Pi r^3\rho}{4\mu^2}
\end{equation}
where $-\Pi$ is the streamwise pressure gradient, $r$ the radius, $\rho$ the density and $\mu$ the viscosity. A state-space system is formed as in Eq.~(\ref{ssrealisation}). For the application under consideration, the state-space system ouput is the velocity field and the input is the volume forcing provided by the nonlinearity. As such the system is square and $f(t)^*v(t)$ has units of power and provides us with the inner product of Eq.~(\ref{innerprod}). For this and most other fluids examples, $D=0$.

The flow in the pipe in question has $Re=2000$. Turbulent behaviour has been observed in pipes at this Reynolds number. The perturbation system studied in this section is axially symmetric ($k=0$) and has azimuthal wavenumber $n=1$. The sector bound used is fairly tight at $[-2830,~14200]$. The infinity norm of the system is $||G||_{\infty}<4800$.

By construction, since the reduced system remains in the same sector as the full linearised model, the supply rate inequality is preserved, as is the nonlinear bound on energy production.

To estimate how many modes should be kept in a reduced-order model, the 
singular values of the balanced ARE solutions are examined in Fig.~\ref{singularvalsk0n1}.
The figure demonstrates that of the $2N=160$ states, the contribution to the supply rate of the fourth mode is approximately two decades less than that of the leading mode.
For comparison the Hankel singular values found during canonical balanced truncation are also shown.

The first and third modes are given in Figs.~\ref{mode1k0n1} and \ref{mode3k0n1}. The tenth mode is shown in Fig.~\ref{mode10k0n1}.
The value of $\sigma$ for the tenth mode suggests that this and similar modes may be involved in the transfer of energy between the modes rather than energy production. Comparison between all modes shows that the most active modes tend to be simpler, because dissipation is more significant for the more complex modes where significant mixing is evident.

\begin{figure}[p]
\centering \includegraphics[width=12cm]{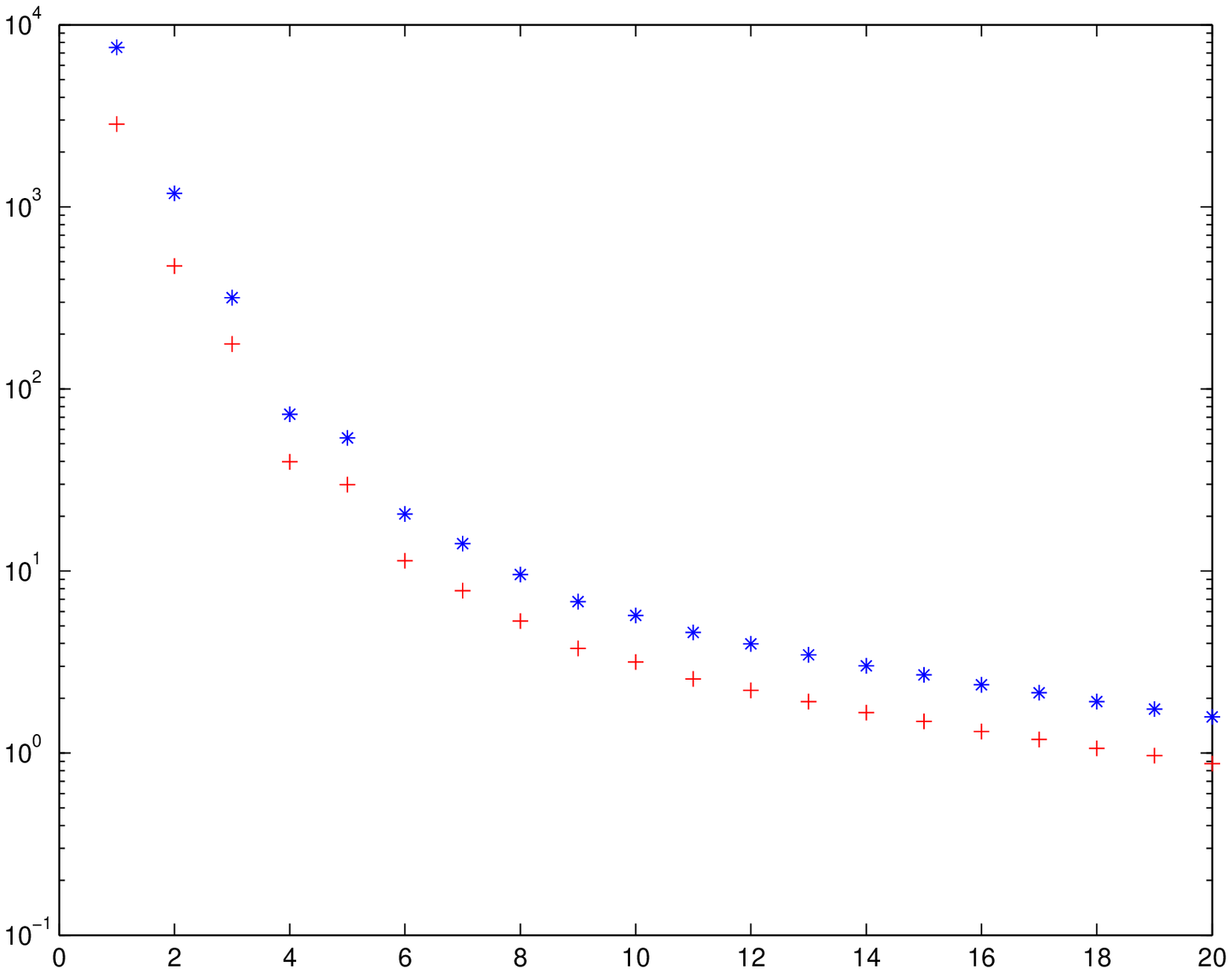}
\caption{The first twenty singular values of the balanced ARE solutions for sector balanced truncation for sector $[-2830,~14200]$ (\textcolor{blue}{*}) and the Hankel singular values (\textcolor{red}{+}).}
\label{singularvalsk0n1}
\end{figure}

\begin{figure}[p]
\centering \includegraphics[width=12cm]{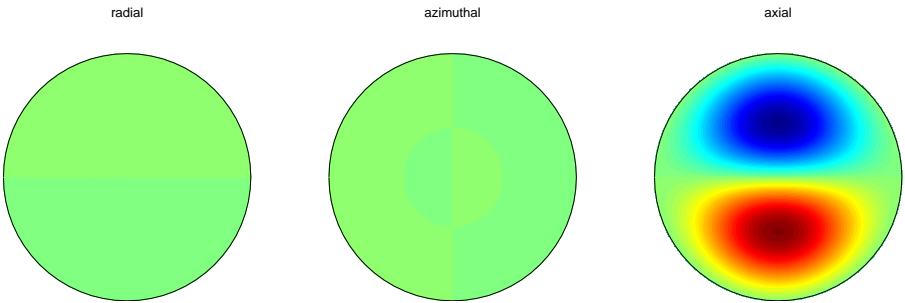}
\caption{$Re=2000$, $k=0$, $n=1$. The velocity field of the most important mode, with $\sigma = 7518$}
\label{mode1k0n1}
\end{figure}
\begin{figure}[p]
\centering \includegraphics[width=12cm]{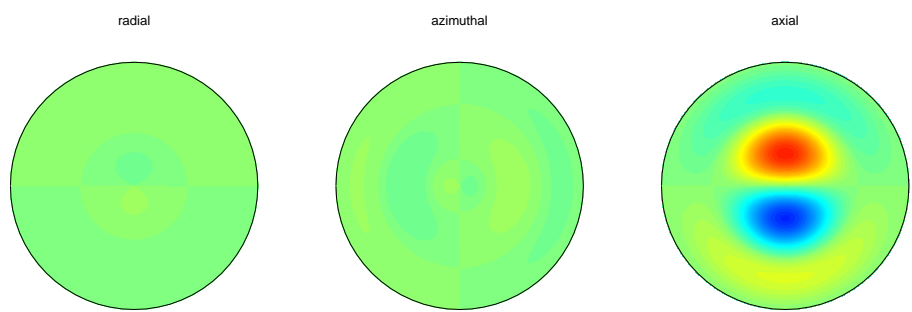}
\caption{$Re=2000$, $k=0$, $n=1$. The third most important mode, with $\sigma = 317$}
\label{mode3k0n1}
\end{figure}
\begin{figure}[p]
\centering \includegraphics[width=12cm]{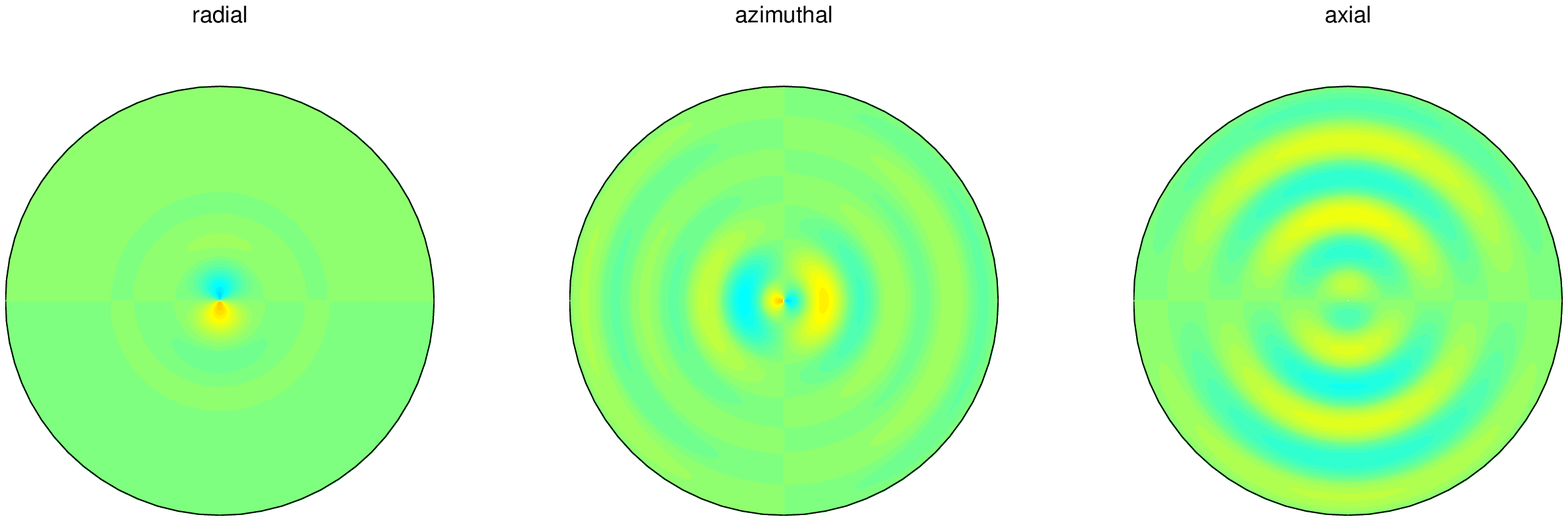}
\caption{$Re=2000$, $k=0$, $n=1$. The tenth mode, with $\sigma = 6$}
\label{mode10k0n1}
\end{figure}

For the purposes of comparison with canonical balanced truncation, a very aggressive model reduction keeping just a single state was calculated using sector-balanced truncation and canonical balanced truncation. The infinity norms for the reduced models were $5710$ for balanced truncation and $4836$ for sector-balanced truncation. The sector-balanced truncation happens to outperform in this case on a traditional infinity-norm measure, although of course the method gives no \emph{a priori} bounds on this norm, unlike canonical balanced truncation.
More pertinent to our purpose is the effect of the reduction procedure on the supply rate. This was evaluated directly by examining the eigenvalues of Eq.~(\ref{matrixsectordefn}) for the original and reduced models using the sector bound $[-2830,~14200]$ over a range of frequencies. This comparison in shown in Fig.~\ref{supplyrateerror}. In the figure, the lower panel shows the fractional error of the reduced models compared with the original, defined by $\left( \lambda_i (\hat{\Lambda}(i\omega)) - \lambda_i (\Lambda(i\omega)) \right) / \lambda_i (\Lambda(i\omega))$. The matching of the supply rate using the sector bounded method is evidently superior. In fact, the balanced truncation method breaks the sector bounds, significantly violating the inequality of Eq.~(\ref{matrixsectordefn}). However, it should be noted that with a higher-order reduced model, both methods give very good results.
\begin{figure}[p]
\centering \includegraphics[width=12cm]{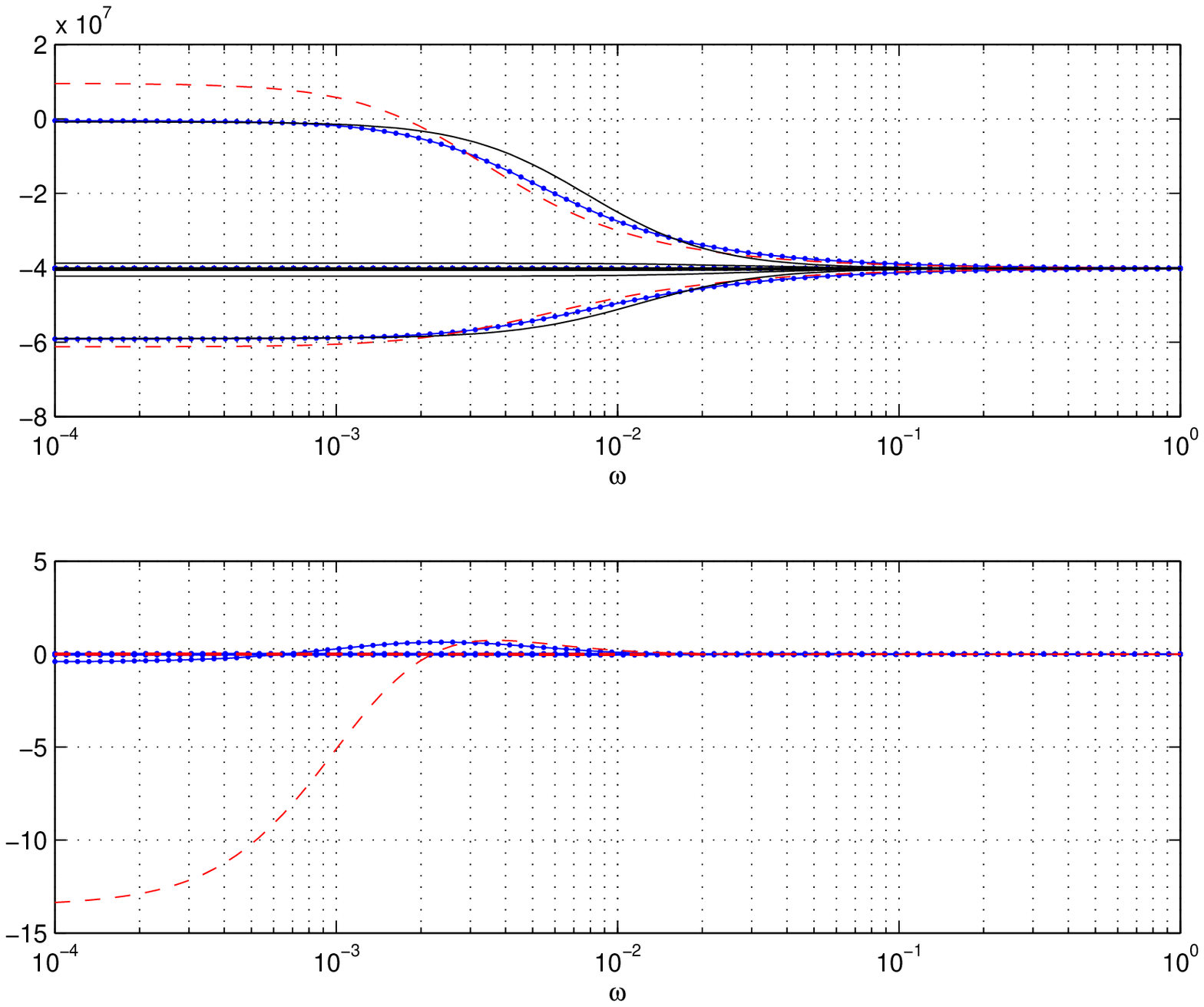}
\caption{$Re=2000$, $k=0$, $n=1$. The upper panel shows the eigenvalues of  Eq.~(\ref{matrixsectordefn}) for the original model (-), the sector-balanced reduced model (\textcolor{blue}{-$\cdot$-}) and for balanced truncation (\textcolor{red}{-~-}). The balanced truncation method violates the inequality of Eq.~(\ref{matrixsectordefn}) showing that the reduced model is no longer in the sector chosen. The lower panel shows the fractional error of the reduced models compared with the original.}
\label{supplyrateerror}
\end{figure}

\section{Conclusions}
A Riccati-based method for model reduction based on balancing with regard to sector bounding has been presented. The supply rate of the reduced model approximates the supply rate of the original model. When applied to the linearised Navier-Stokes equations and considering the passivity of the nonlinear term neglected in the linearisation, this sector bounding results in a bound on the turbulent energy production of the full nonlinear Navier-Stokes equations at a given Reynolds number. These bounds may be calculated from the linearised system.
Because the scheme preserves the sector bounds of the linearised system it also preserves bounds on the turbulent energy production of the whole system, including the nonlinearity.

The scheme has been applied to linearised perturbations about pipe flow for a typical case of $Re=2000$. The perturbation system is axially symmetrix and has azimuthal wavenumber of $1$. Representative modes coresponding to energy production and transfer have been identified and presented.
A comparison shows that sector-balanced truncation approximates the supply rate at least as well as, or better than, canonical balanced truncation,


\section*{Acknowledgments}
  The author wishes to thank the UK EPSRC for their support.
\bibliography{refs} 
\bibliographystyle{ijbc}

\appendix
\input{appendix1}

\end{document}

%% file: header.tex
\def\BibTeX{{\rm B\kern-.05em{\sc i\kern-.025em b}\kern-.08em
    T\kern-.1667em\lower.7ex\hbox{E}\kern-.125emX}}

\def\Hinf{$\cal {H}_\infty~$}

\def\Hinfty{{{\cal{H}}_{\infty}}}

\def\grad{\V{\bigtriangledown}}

\def\div{\grad \cdot }

\def\s={\stackrel{s}{=}}

\newcommand{\curly}[1]{{\cal{#1}}}

\newcommand{\V}[1]{{\boldsymbol{#1}}}

\newcommand{\bspace}[1]{{\mathbb{#1}}}

\newcommand{\norm}[2]{\lVert#1\rVert_{#2}}
\newcommand{\inprod}[2]{\left<{#1},{#2}\right>}

\newcommand{\mat}[2]{\left[\begin{array}{#1} #2\end{array}\right]}

%% file: appendix1.tex
\section{A sector balancing transformation}
\label{apx_balancing}
This appendix concerns the stable, finite-dimensional, linear time-invariant system $v=Hf$ defined in Eq. \ref{ssrealisation}.

\begin{lem}
There exists a similarity transformation $T$ so that the solutions to the two sector bounding algebraic Riccati equations (AREs) for the realisation in Eqs.~(\ref{transformed1}-\ref{transformed}) are equal and diagonal with ordered elements.
The transformation that achieves this is $T=\Sigma^{\frac{1}{2}}U'R^{-1}$ where $Y=RR'$ is a Cholesky factorisation of $Y$ and $R'XR = U\Sigma^{2}U'$ is a singular value decomposition of $R'XR$.
\end{lem}
\begin{proof}Since the system $H$ is in sector $[a,b]$, so is the transformed system in Eq.~\ref{transformed}.
Forming the appropriate ARE analogous to Eq.~\ref{RiccX} for the transformed system results in a new ARE with solution $X_T$,
\begin{eqnarray}
\lefteqn{{T^{-1}}'A'T'X_T + X_T TAT^{-1} -} \nonumber \\
& & (X_T TB-{T'}^{-1}C'(\alf I - D))P^{-1}(X_TB-{T'}^{-1}C'(\alf I - D))' + \nonumber \\
& & {T'}^{-1}C'CT^{-1}=0 
\label{RiccX_T}
\end{eqnarray}
where $P = D'D - \alf(D+D') + abI$.

Substituting $RR'$ for $Y$ and $R'^{-1}U\Sigma^{2}U'R^{-1}$ for $X$ into Eq.~\ref{RiccX} gives
\begin{eqnarray}
\lefteqn{A'R'^{-1}U\Sigma^{2}U'R^{-1} + R'^{-1}U\Sigma^{2}U'R^{-1}A -} \nonumber \\
& & (R'^{-1}U\Sigma^{2}U'R^{-1}B-C'(\alf I - D))P^{-1}(R'^{-1}U\Sigma^{2}U'R^{-1}B-C'(\alf I - D))' + \nonumber \\
& &  C'C=0.
\end{eqnarray}
Left-multiplying by $T'^{-1}=\Sigma^{-1/2}U'R'$ and right-multiplying by $T^{-1}=RU\Sigma^{-1/2}$ results in
\begin{eqnarray}
\lefteqn{\left[\Sigma^{-1/2}U'R'A'R'^{-1}U\Sigma^{1/2}\right]\Sigma + \Sigma\left[\Sigma^{1/2}U'R^{-1}ARU\Sigma^{-1/2}\right] -}\nonumber \\
& &  (\Sigma\Sigma^{1/2} U'R^{-1}B-\Sigma^{-1/2}U'R'C'(\alf I - D))P^{-1}(\Sigma\Sigma^{1/2}U'R^{-1}B-\Sigma^{-1/2}U'R'C'(\alf I - D))' + \nonumber \\
& & \Sigma^{-1/2}U'R'C'CRU\Sigma^{-1/2}=0. 
\label{RiccX_Sigma}
\end{eqnarray}
Comparison of Eqs.~\ref{RiccX_Sigma} and \ref{RiccX_T} shows that with $T=\Sigma^{\frac{1}{2}}U'R^{-1}$, $\Sigma$ is a solution of Eq.~\ref{RiccX_T}.
The proof for the second ARE analogous to Eq.~\ref{RiccY} proceeds similarly.
Since $H\in[a,b]$ regardless of realisation, $\Sigma$ is a stabilising solution to both AREs.
\end{proof}

\begin{lem}The truncated realisation in Eq.~(\ref{truncated}) is itself sector balanced, stable, and $H_{r}$ is also strictly in sector $[a,b]$.
\end{lem}
\begin{proof}If $X=\mat{cc}{X_1 & 0 \\ 0 & X_2}>0$ then direct substitution into Eq. \ref{RiccX} shows that 
\begin{equation}
A_{11}'X_1 + X_1 A_{11} - (X_1 B_1 -C_1'(\alf I - D))P^{-1}(X_1 B_1-C_1'(\alf I - D))' + C_1'C_1=0
\label{RiccX_1}
\end{equation}
and similarly for $X_2$.

Furthermore, direct substitution of $X=\mat{cc}{X_1 & 0 \\ 0 & X_2}$ into Eq. \ref{StabilisingX} gives
\begin{eqnarray}
\lefteqn{\mat{cc}{A_{11} & A_{12} \\ A_{21} & A_{22}} + \mat{cc}{B_1 P^{-1}(\alf I-D') C_{1} & B_1 P^{-1}(\alf I-D') C_{2} \\ B_2 P^{-1}(\alf I-D') C_{1} & B_2 P^{-1}(\alf I-D') C_{2}} -} \nonumber \\
& & \mat{cc}{B_1 P^{-1}B_1'X_1 & B_1 P^{-1}B_2'X_2 \\ B_2 P^{-1}B_1'X_1  & B_2 P^{-1}B_2'X_2 } <0.
\end{eqnarray}
Considering this as of the form \[\mat{cc}{\alpha & \beta \\ \gamma & \delta}<0\] means that 
\[\mat{cc}{x_1' & x_2'} \mat{cc}{\alpha & \beta \\ \gamma & \delta} \mat{c}{x_1 \\ x_2} <0 \quad \forall x,~x= \mat{c}{x_1 \\ x_2}.\]
Since this is true for all $x$ it must be true for the case $x_2=0$. Therefore
\begin{equation}
A_{11} - B_1 P^{-1}(\alf I-D)'C_1 - B_1P^{-1}B_1'X_1<0
\end{equation}
so $X_1$ itself is a stabilising solution of the reduced ARE~(\ref{RiccX_1}) (and similarly for $X_2$). A similar argument holds for the second ARE~(\ref{RiccY}). The result then follows immediately taking $X=\Sigma$.
\end{proof}

%% file: A_Sharma_3581-250108.bbl
\begin{thebibliography}{}

\bibitem[Bamieh and Dahleh, 2001]{Bamieh01}
Bamieh, B. and Dahleh, M. (2001).
\newblock Energy amplification in channel flows with stochastic excitation.
\newblock {\em Phys. Fluids}, 13(11).

\bibitem[Bewley, 2001]{Bewley01}
Bewley, T. (2001).
\newblock Flow control: New challenges for a new renaissance.
\newblock {\em Progress in Aerospace Sciences}, 37:21--58.

\bibitem[Farrell and Ioannou, 1993]{Farrell93}
Farrell, B. and Ioannou, J. (1993).
\newblock Stochastic forcing of the linearized {N}avier-{S}tokes equations.
\newblock {\em Phys. Fluids}, 5(11):2600--2609.

\bibitem[Green and Limebeer, 1995]{Green}
Green, W.~J. and Limebeer, D.~J.~N. (1995).
\newblock {\em Linear Robust Control}.
\newblock Prentice Hall, New Jersey.

\bibitem[Gugercin and Antoulas, 2004]{Gugercin04}
Gugercin, S. and Antoulas, A.~C. (2004).
\newblock A survey of model reduction by balanced truncation and some new
  results.
\newblock {\em Int. Journal of Control}, 77(8):748--766.

\bibitem[Jovanovi\'{c} and Bamieh, 2005]{Jovanovic05}
Jovanovi\'{c}, M.~R. and Bamieh, B. (2005).
\newblock Componentwise energy amplification in channel flows.
\newblock {\em J. Fluid Mech.}, 534.

\bibitem[Meseguer and Trefethen, 2003]{Meseguer03-1}
Meseguer, A. and Trefethen, L.~N. (2003).
\newblock Linearized pipe flow to {R}eynolds number $10^7$.
\newblock {\em Journal of Computational Physics}, (186):178--197.

\bibitem[Moore, 1981]{Moore81}
Moore, B.~C. (1981).
\newblock Principal component analysis in linear systems: Controllability,
  observability, and model reduction.
\newblock {\em IEEE Trans. Automat. Contr}, 26(1):17--32.

\bibitem[Rowley, 2004]{Rowley04}
Rowley, C.~W. (2004).
\newblock Model reduction for fluids, using balanced proper orthogonal
  decomposition.
\newblock {\em Int. J. on Bifurcation and Chaos}, 15(3):997--1013.

\bibitem[Schmid, 2007]{Schmid07}
Schmid, P.~J. (2007).
\newblock Nonmodal stability theory.
\newblock {\em Ann. Rev. Fluid Mech.}, (39):129--162.

\bibitem[Sharma et~al., 2006]{Sharma06AIAA}
Sharma, A.~S., McKeon, B.~J., Morrison, J.~F., and Limebeer, D.~J.~N. (2006).
\newblock Stabilising control laws for the incompressible {N}avier-{S}tokes
  equations using sector stability theory.
\newblock {\em 3rd AIAA Flow Control Conference}.

\bibitem[Shim, 1997]{Shim97}
Shim, D.-S. (1997).
\newblock Synthesis of sector-bounded {LTI} systems.
\newblock Proceedings of the 36th Conference on Decision and Control, San
  Diego, California, USA.

\bibitem[Trefethen et~al., 1993]{Trefethen93}
Trefethen, L.~N., Trefethen, A.~E., Reddy, S., and Driscoll, T.~A. (1993).
\newblock Hydrodynamic stability without eigenvalues.
\newblock {\em Science}, 261(5121).

\bibitem[Waleffe, 2003]{Waleffe03}
Waleffe, F. (2003).
\newblock Homotopy of exact coherent structures in plane shear flows.
\newblock {\em Phys. Fluids}, 15(6).

\bibitem[Wedin and Kerswell, 2004]{Wedin04}
Wedin, H. and Kerswell, R.~R. (2004).
\newblock Exact coherent structures in pipe flow: Travelling wave solutions.
\newblock {\em J. Fluid Mech.}, 508:333--371.

\bibitem[Willems, 1972a]{Willems72-0}
Willems, J.~C. (1972a).
\newblock Dissipative dynamical systems part 1: General theory.
\newblock {\em Arch. Rational Mech. Anal.}, 45:321--351.

\bibitem[Willems, 1972b]{Willems72-1}
Willems, J.~C. (1972b).
\newblock Dissipative dynamical systems part 2: Linear systems with quadratic
  supply rates.
\newblock {\em Arch. Rational Mech. Anal.}, 45:352--393.

\bibitem[Zames, 1966]{Zames66-1}
Zames, G. (1966).
\newblock On the input-output stability of time-varying nonlinear feedback
  systems - {P}art {I}: Conditions derived using concepts of loop gain,
  conicity and positivity.
\newblock {\em IEEE Trans. on Automatic Control}, AC-11(2):228--238.

\end{thebibliography}
